\def\year{2019}\relax
\documentclass[letterpaper]{article} 
\usepackage{aaai19}  
\usepackage{times}  
\usepackage{helvet}  
\usepackage{courier}  
\usepackage{url}  
\usepackage{graphicx}  
\usepackage{natbib}
\usepackage{amsmath,amsthm,amssymb,algorithm,algpseudocode,graphicx,tikz,caption,subcaption,xcolor}
\definecolor{light-gray}{gray}{0.7}

\algrenewcommand\algorithmicindent{0.43em}
\newtheorem{theorem}{Theorem}[section]
\newtheorem{proposition}[theorem]{Proposition}
\newtheorem{lemma}[theorem]{Lemma}
\algnotext{EndFor}
\algnotext{EndIf}
\algnotext{EndWhile}
\algnotext{EndProcedure}

\theoremstyle{definition}
\newtheorem{definition}[theorem]{Definition}

\DeclareMathOperator*{\argmax}{arg\,max}
\DeclareMathOperator*{\E}{\mathbb{E}}

\nocopyright
\allowdisplaybreaks

\begin{document}
%
\title{When Do Envy-Free Allocations Exist?}
\author{Pasin Manurangsi\\
Department of EECS\\
UC Berkeley
\And 
Warut Suksompong\\
Department of Computer Science\\
University of Oxford}
\maketitle
\begin{abstract}
We consider a fair division setting in which $m$ indivisible items are to be allocated among $n$ agents, where the agents have additive utilities and the agents' utilities for individual items are independently sampled from a distribution. Previous work has shown that an envy-free allocation is likely to exist when $m=\Omega(n\log n)$ but not when $m=n+o(n)$, and left open the question of determining where the phase transition from non-existence to existence occurs. We show that, surprisingly, there is in fact no universal point of transition---instead, the transition is governed by the divisibility relation between $m$ and $n$. On the one hand, if $m$ is divisible by $n$, an envy-free allocation exists with high probability as long as $m\geq 2n$. On the other hand, if $m$ is not ``almost'' divisible by $n$, an envy-free allocation is unlikely to exist even when $m=\Theta(n\log n/\log\log n)$.
\end{abstract}

\section{Introduction}

Resource allocation is a fundamental task that occurs in a great number of everyday situations, from allocating school supplies to children and course slots in universities to students, to allocating machine processing time to users and kidneys to kidney transplant patients. One of the principal concerns when allocating resources to interested agents is \emph{fairness}: we want all agents to feel that they receive a fair share of the resources. There is a rich and beautiful theory of \emph{fair division} that goes back several decades and has been studied in mathematics, economics, and more recently in computer science \citep{BramsTa96,Moulin03}.

In order to reason about fairness, we must define when an allocation is considered to be ``fair''. One of the most prominent fairness notions is \emph{envy-freeness}, which means that every agent likes her allocated portion at least as much as that of any other agent \citep{Foley67,Varian74}. While an envy-free allocation can always be obtained when we allocate \emph{divisible} goods such as land or machine processing time \citep{Stromquist80}, this is not the case when it comes to allocating \emph{indivisible} goods like jewelry and artworks. Indeed, if a single bracelet or painting is to be divided between two agents, then no matter how the division is performed, the agent who does not receive the item will be left envying the other agent.

Given that the existence of envy-free allocations cannot be guaranteed in general for indivisible goods, an important question is therefore \emph{when} such allocations exist. \cite{DickersonGoKa14} investigated this question under a simple model where the agents have additive utilities and their utilities for individual items are drawn at random from probability distributions. If the number of items, $m$, is less than the number of agents, $n$, no envy-free allocation exists since any allocation necessarily leaves some agent empty-handed and envious. Dickerson et al. showed that even when the number of items slightly exceeds the number of agents---$m=n+o(n)$---an envy-free allocation is still unlikely to exist. However, as soon as the number of items is larger than the number of agents by a logarithmic factor---$m=\Omega(n\log n)$---an envy-free allocation exists with high probability, and can furthermore be obtained by simply giving each item to the agent with the highest utility for it. Dickerson et al. also found the phase transition from non-existence to existence to be quite sharp in computer experiments, and left open the question of determining where this transition occurs. Is the logarithmic factor in the upper bound necessary, or do we already have existence when, say, $m=1.001n$?

In this paper, we show that, surprisingly, there is in fact \emph{no} universal point of transition between non-existence and existence. Instead, the transition is governed by the divisibility relation between $m$ and $n$. On the one hand, if $m$ is divisible by $n$, we show that an envy-free allocation exists with high probability as long as $m\geq 2n$ (Theorem~\ref{thm:existence}). Our result improves upon the aforementioned $m=\Omega(n\log n)$ upper bound and moreover completely closes the gap for the case of divisibility, since Dickerson et al.'s lower bound already implies that the same result does not hold when $m=n$.\footnote{We do note, however, that Dickerson et al.'s upper bound holds under a weaker assumption on the distributions. For example, it does not assume that the utilities are drawn from the same distribution for all agents and items.} On the other hand, if $m$ is not ``almost'' divisible by $n$, in the sense that the remainder of the division is between $n^\epsilon$ and $n-n^\epsilon$ for some constant $\epsilon\in(0,1)$, we show that an envy-free allocation is unlikely to exist as long as $m=O(n\log n/\log\log n)$ (Theorem~\ref{thm:non-existence}). This comes to within a $\Theta(\log\log n)$ factor of matching their upper bound. Both our existence and non-existence results rely on several new key ideas. In particular, for the existence result we need a completely different algorithm, since the welfare-maximizing algorithm used to achieve existence for $m=\Omega(n\log n)$ cannot yield any improvement of this bound (Proposition~\ref{prop:coupon}). 

\subsection{Related Work}
\label{sec:relatedwork}

Besides the work of \cite{DickersonGoKa14} that we mentioned, several other works have investigated the asymptotic existence and non-existence of fair allocations for various fairness notions. \cite{Suksompong16} considered \emph{proportional} allocations---allocations in which every agent receives at least $1/n$ of her value for the whole set of items---and showed that such allocations exist with high probability if $m$ is a multiple of $n$ or $m=\omega(n)$. \cite{KurokawaPrWa16} showed that an allocation that satisfies the \emph{maximin share criterion} is likely to exist as long as either $m$ or $n$ goes to infinity.\footnote{We refer to their paper for the definition, but remark here that both proportionality and the maximin share criterion are weaker than envy-freeness when utilities are additive.} As in our work, both \cite{KurokawaPrWa16} and \cite{Suksompong16} used techniques from the theory of matchings in random graphs to establish the existence of fair allocations. \cite{AmanatidisMaNi17} also addressed the existence of allocations satisfying the maximin share criterion. Finally, \cite{ManurangsiSu17} considered the setting where goods are allocated to groups of agents and generalized \cite{DickersonGoKa14}'s results on envy-freeness to that setting.

Since envy-free allocations cannot always be obtained even in the simplest setting with two agents and one item, a recent line of work has focused on relaxations of envy-freeness with the goal of recovering the guaranteed existence. These relaxations include \emph{envy-freeness up to one good}---any envy that an agent has towards another agent can be eliminated by removing \emph{some} item from the latter agent's bundle---and \emph{envy-freeness up to any good}---any such envy can be eliminated by removing \emph{any} item from the latter agent's bundle. It has been shown that these relaxations do provide existence guarantees in a number of settings \citep{LiptonMaMo04,CaragiannisKuMo16,BiswasBa18,PlautRo18}.

\section{Preliminaries}
\label{sec:prelims}

A set $M=[m]$ of indivisible items is to be allocated to a set $N=[n]$ of agents, where we use $[k]$ to denote the set $\{1,2,\dots,k\}$. Each agent $i$ has a nonnegative utility $u_i(j)$ for item $j$. We assume that the utility $u_i(j)$ lies in $[0,1]$; this does not introduce a loss of generality since we can scale down all utilities by their maximum. The utilities of the agents are \emph{additive}, i.e., $u_i(M')=\sum_{j\in M'}u_i(j)$ for any $M'\subseteq M$. The additivity assumption is made in several works on fair division and, in particular, in all of the works on the asymptotic existence of fair allocations mentioned in Section~\ref{sec:relatedwork}.

A \emph{bundle} refers to a subset of $M$. An \emph{allocation} is a partition of $M$ into $n$ bundles $(M_1,M_2,\dots,M_n)$, where bundle $M_i$ is allocated to agent $i$. An allocation is said to be \emph{envy-free for agent $i$} if $u_i(M_i)\geq u_i(M_j)$ for any $j\in N$, and \emph{envy-free} if it is envy-free for every agent $i\in N$.

For agents $i\in N$ and items $j\in M$, the utilities $u_i(j)$ are drawn independently from a distribution $\mathcal{U}$. A distribution is said to be \emph{non-atomic} if it does not put positive probability on any single point. The condition that we will impose on $\mathcal{U}$ for our results is that it ``behaves like a polynomial close to 1'' in the sense that the function $g(\alpha) = \Pr_{u \sim \mathcal{U}}[u \geq 1 - \alpha]$ is bounded above and below by a polynomial. This is formalized in the following definition. 

\begin{definition}
\label{def:polybounded}
Let $\theta, q$ be any positive real numbers. A probability distribution $\mathcal{U}$ on $[0, 1]$ is said to be \emph{$(\theta, q)$-polynomially bounded below} (resp. \emph{above}) \emph{at 1} if for every $\alpha \in (0, 1]$, we have $\Pr_{u \sim \mathcal{U}}[u > 1 - \alpha] \geq \theta \cdot \alpha^q$ (resp. $\Pr_{u \sim \mathcal{U}}[u > 1 - \alpha] \leq \theta \cdot \alpha^q$).

A probability distribution $\mathcal{U}$ is said to be \emph{polynomially bounded at 1} if there exist constants $\underline{\theta}, \overline{\theta}, q > 0$ such that $\mathcal{U}$ is $(\underline{\theta}, q)$-polynomially bounded below at 1 and $(\overline{\theta}, q)$-polynomially bounded above at 1.
\end{definition}

We assume in Section~\ref{sec:existence} that $\mathcal{U}$ is polynomially bounded at 1 and in Section~\ref{sec:non-existence} that $\mathcal{U}$ is polynomially bounded below at 1. To illustrate the generality of this definition, consider any non-atomic continuous distribution $\mathcal{U}$ whose probability density function $f_{\mathcal{U}}$ is bounded below (resp. above) around 1, i.e., there exist $\epsilon, \beta > 0$ such that $f_{\mathcal{U}}(x) \geq \beta$ (resp. $f_{\mathcal{U}}(x) \leq \beta$) for all $x \geq 1 - \epsilon$. One can check that $\mathcal{U}$ is polynomially bounded below (resp. above) at 1 with parameters $\theta = \epsilon \cdot \beta$ (resp. $\theta = \max\{\beta, 1/\epsilon\}$) and $q = 1$. This immediately implies that the uniform distribution on $[0, 1]$ and a normal distribution (with any mean and variance) truncated at 0 and 1 are polynomially bounded at 1, as both have probability density functions that are bounded both above and below in $[0, 1]$.

For completeness, let us also provide examples of distributions that are not polynomially bounded at 1. The first example is when $\Pr_{u \sim \mathcal{U}}[u = 1] > 0$. In this case, clearly the distribution is not polynomially bounded above at 1. Another example is if we take any $\mathcal{U}$ such that $\Pr_{u \sim \mathcal{U}}[u \geq 1 - 1/2^i] = 1/2^{i^2}$ for all integer $i \geq 0$. It is not hard to see that this distribution is not polynomially bounded below at 1. Indeed, for any fixed $q > 0$, we have $\lim_{i \to \infty} \frac{(1/2^{i^2})}{(1/2^i)^q} = 0$, which means that there is no $\theta > 0$ such that $\Pr_{u \sim \mathcal{U}}[u \geq 1 - \alpha] \geq \theta \cdot \alpha^q$ for all $\alpha \in (0, 1]$.

Finally, a statement is said to hold \emph{with high probability} if the probability that it holds approaches 1 as $n\rightarrow\infty$.

\section{Existence}
\label{sec:existence}

In this section, we investigate the existence front of envy-free allocations. We first show that the welfare-maximizing algorithm of \cite{DickersonGoKa14} cannot yield any improvement of the $m=\Omega(n\log n)$ bound. We then prove the main existence result of this paper, which holds for any $m\geq 2n$ that is a multiple of $n$: 

\begin{theorem}
\label{thm:existence}
Let $r\geq 2$ be an integer, and suppose that $m=rn$. Assume that $\mathcal{U}$ is polynomially bounded at 1. With high probability, there exists an envy-free allocation. Moreover, there is a polynomial-time algorithm that computes such an allocation.
\end{theorem}

A bonus of our algorithm is that it returns a \emph{balanced} allocation, i.e., one that gives every agent the same number of items. This may be desirable in situations where capacity constraints are involved, for example if we divide artworks between museums or players between sports teams.

\subsection{The Limit of the Welfare-Maximizing Algorithm}

Recall the main existence result of \cite{DickersonGoKa14}: when $m = \Omega(n\log n)$, the welfare-maximizing algorithm, which allocates each item to the agent who values it most, is likely to produce an envy-free allocation. We observe that this bound is tight up to a constant factor---for $m = n \log n - \omega(n)$ items, the welfare-maximizing allocation is unlikely to be envy-free. An implication of this observation is that the welfare-maximizing algorithm fails to be envy-free in the case where $m = rn$, for any positive integer $r \leq \log n-\omega(1)$. By contrast, the algorithm that we will present finds an envy-free allocation with high probability for any integer $r \geq 2$.

\begin{proposition}
\label{prop:coupon}
Let $m = n \log n - \omega(n)$, and suppose that $\mathcal{U}$ is non-atomic. Then, with high probability, the welfare-maximizing allocation is not envy-free.
\end{proposition}

\begin{proof}
The proposition follows from a classical result on the \emph{coupon collector's problem}. In this problem, there is an urn of $n$ coupons. Each turn, a coupon is drawn uniformly at random from the urn and immediately returned to the urn. \cite{ErdosRe61} proved that with high probability, after $n \log n - \omega(n)$ turns, some coupon has not been drawn. 

The connection between the coupon collector's problem and our setting is fairly simple. First, the non-atomic assumption on the distribution implies that, almost surely, all items yield positive utility to every agent, and every item has only one agent who values it most. As a result, the welfare-maximizing allocation assigns each item to each agent with probability $1/n$. If we view each agent as a coupon in the coupon collector's problem, Erd{\H{o}}s and R{\'{e}}nyi's result implies that with high probability, some agent does not receive any item in this allocation. From the positive utility observation, the allocation cannot be envy-free.
\end{proof}

\subsection{Warm-Up: A Simplified Algorithm for $r \geq 3$}
\label{sec:simplified}

The remainder of Section~\ref{sec:existence} is devoted to proving Theorem~\ref{thm:existence}; we assume throughout that $m=rn$ for some integer $r\geq 2$. As \cite{DickersonGoKa14} already showed that the theorem holds for $r=\Omega(\log n)$, it suffices for us to establish the statement for $r = O(\log n)$. Nevertheless, we will prove the statement for $r \leq e^{n^{0.1}}$, which is much stronger; while this is not necessary, we do so to demonstrate that our algorithm and its analysis are robust and apply even when the number of items is significantly larger than the number of agents. 

Before we proceed to the actual algorithm, let us provide the intuition behind the algorithm by describing a simpler algorithm that works in all cases except when $r = 2$. For the sake of exposition, we shall restrict ourselves to the case where the distribution $\mathcal{U}$ is the uniform distribution on $[0, 1]$ and $r > 2$ is a constant (i.e., does not grow with $n$). We shall also sometimes be informal here; all proofs will be formalized in the rest of Section~\ref{sec:existence}.

The simplified algorithm tries to find an allocation that satisfies the following two properties: (i) each agent receives exactly $r$ items, and (ii) each agent has utility at least $\tau := 1 - 2\log n/n$ for every item that she receives. If at least one such allocation exists, the algorithm outputs any of them. Else, it outputs NULL. Note that determining whether such an allocation exists and finding one if it exists can be done in polynomial time by reducing to matching: we create a bipartite graph $(N \times [r], M, E)$, where $((i, \ell), j) \in E$ if and only if $u_i(j) \geq \tau$. A desired allocation corresponds to a perfect matching in this graph.\footnote{We write $(U,V,E)$ to denote a bipartite graph with the set of vertices $U$ and $V$ in the partition, which we refer to as the set of \emph{left vertices} and \emph{right vertices} respectively, and the set of edges $E$.}

For the sake of convenience, we introduce the notion of \emph{$r$-matching}, which allows us to focus on the graph with vertex set $N$ instead of $N \times [r]$. In an $r$-matching, each left vertex can be matched to as many as $r$ right vertices, whereas each right vertex is still allowed to be matched to at most one left vertex.

\begin{definition}
An \emph{$r$-matching} of a bipartite graph $G$ is a subgraph of $G$ such that every left vertex has degree at most $r$ and every right vertex has degree at most 1. An $r$-matching is said to be \emph{perfect} if every left vertex has degree exactly $r$ and every right vertex has degree exactly 1.
\end{definition}

As with normal matchings, a perfect $r$-matching can be computed in polynomial time by creating $r$ copies of each left vertex and finding a perfect matching. With this definition, our simplified algorithm can be described as follows.

\begin{algorithm}
\caption{Simplified Algorithm for $r\geq 3$}\label{alg:matching-basic}
\begin{algorithmic}[1]
\Procedure{ThresholdMatching$_\tau(N, M, \{u_i\}_{i\in[n]})$}{}
\For{$i = 1,2,\dots, n$}
\State $M_{\geq \tau}(i) \leftarrow \{j \in M \mid u_i(j) \geq \tau\}$
\EndFor 
\State Let $G_{\geq \tau} = (N, M, E_{\geq \tau})$ denote the graph where $(i, j) \in E_{\geq \tau}$ iff $j \in M_{\geq\tau}(i)$.
\If{$G_{\geq \tau}$ contains a perfect $r$-matching}
\State \Return any perfect $r$-matching of $G_{\geq \tau}$
\Else
\State \Return NULL
\EndIf
\EndProcedure
\end{algorithmic}
\end{algorithm}

We now sketch the proof of correctness of Algorithm~\ref{alg:matching-basic}, which consists of two parts. Firstly, we argue that with high probability, the algorithm returns a perfect $r$-matching in $G_{\geq \tau}$ (i.e., does not output NULL). Secondly, we show that the output allocation is envy-free with high probability. 

\subsubsection{Existence of a perfect $r$-matching in $G_{\geq \tau}$.} 

For the first part, we evoke a classical result regarding the existence of a perfect matching in bipartite random graphs. Recall that for any positive integers $a,b$ and any $p\in[0,1]$, a bipartite graph sampled from the \emph{Erd{\H{o}}s-R{\'{e}}nyi random bipartite graph distribution} $\mathcal{G}(a,b,p)$ consists of left and right vertex sets $A$ and $B$ of size $a$ and $b$ respectively, and for any pair of vertices $a\in A$ and $b\in B$, the edge $(a,b)$ occurs with probability $p$ independently of other pairs of vertices.

\begin{proposition}[\cite{ErdosRe64}]
\label{prop:er-matching}
Let $G$ be a bipartite graph sampled from the Erd{\H{o}}s-R{\'{e}}nyi random bipartite graph distribution $\mathcal{G}(n, n, p)$, where $p = (\log n + \omega(1))/n$. Then, with high probability, $G$ contains a perfect matching.
\end{proposition}

To show that a perfect $r$-matching is likely to exist in $G_{\geq \tau}$, we arbitrarily partition the item set $M$ into $r$ parts $M^{(1)}, \dots, M^{(r)}$, each of size $n$. We also create a bipartite graph $H^{(a)}$ for $a = 1,2, \dots, r$ where the left vertex set is $N$, the right vertex set is $M^{(a)}$, and each $(i, j)$ is an edge iff $u_i(j) \geq \tau$. Now, since $\tau=1-2\log n/n$, for each $a$ the graph $H^{(a)}$ is distributed according to the Erd{\H{o}}s-R{\'{e}}nyi random bipartite graph distribution $\mathcal{G}(n, n, 2\log n / n)$. As a result, Proposition~\ref{prop:er-matching} implies that $H^{(a)}$ contains a perfect matching with high probability. By taking the union of the perfect matchings in $H^{(1)}, \dots, H^{(r)}$, we conclude that $G_{\geq \tau}$ contains a perfect $r$-matching with high probability. This completes the first part of the proof sketch.

\subsubsection{Envy-freeness of output allocation.} 

Next, we argue that with high probability, any allocation output by Algorithm~\ref{alg:matching-basic} is envy-free. Consider any such allocation. Since every agent receives $r$ items, each of which yields utility at least $\tau$ to her, her total utility is at least $r \cdot \tau = r - 2r\log n / n$. It therefore suffices to show that with high probability, for every $i' \ne i$, the utility of agent $i'$ for agent $i$'s bundle $M_i$ is at most $r - 2r\log n/n$. We will show that with high probability, for every $i' \ne i$, agent $i'$ values at most $r - 1$ items in $M_i$ more than $1 - 2r\log n/n$. This is sufficient because these $r - 1$ items can each contribute utility at most 1 to agent $i'$, whereas the remaining item contributes utility at most $1 - 2r\log n/n$ to her. It follows that the utility of agent $i'$ for $M_i$ does not exceed $(r-1) + (1-2r\log n/n) = r - 2r\log n/n$.

Fix two distinct $i,i'\in[n]$. Let $E_{i,i'}$ denote the ``bad'' event that there exist $r$ items $j_1,\dots,j_r$ for which $u_i(j_k) \geq \tau$ and $u_{i'}(j_k) \geq 1 - 2r\log n / n$ for $k=1,2,\dots,r$. Consider any item $j \in M$. Since we assume that $u_i(j)$ and $u_{i'}(j)$ are drawn independently from the uniform distribution on $[0,1]$, the probability that item $j$ satisfies the two inequalities above for $i$ and $i'$ is at most $\frac{2\log n}{n} \cdot \frac{2r\log n}{n} = \frac{4r\log^2 n}{n^2}$. Using the union bound over all subsets of $r$ items, we have
\begin{align*}
\Pr[E_{i,i'}] &\leq \binom{m}{r} \cdot \left(\frac{4r\log^2 n}{n^2}\right)^r \\
              &\leq \left(\frac{4r^2\log^2 n}{n}\right)^r = o(n^{-2}),
\end{align*}    
where we use the inequality $\binom{m}{r}\leq m^r$ and the assumption that $r\geq 3$ is constant.
Applying the union bound again over all $i, i'$, the probability that at least one bad event occurs is $o(1)$. This concludes our proof sketch for the simplified algorithm.

\subsection{The Algorithm}

Having described the simplified algorithm, we now proceed to the actual algorithm. Before we do so, let us note that Algorithm~\ref{alg:matching-basic} does \emph{not} work for $r = 2$. This is because when $r = 2$, there is a constant probability that some pair of agents have the same two most valued items. In this case, Algorithm~\ref{alg:matching-basic} could output an allocation that assigns both items to one of the two agents, which would mean that this agent is envied by the other agent.

To make Algorithm~\ref{alg:matching-basic} work for $r = 2$, recall that the algorithm could fail if it is possible to find $r$ items in the candidate item set of agent $i$ (i.e., the set of items for which agent $i$ has utility at least $\tau$) that another agent $i'$ values more than $r \cdot \tau$ in total. The modification to the algorithm is simple: remove any such problematic items from the candidate set of $i$ before we try to find a perfect $r$-matching in the graph.

There are multiple ways to implement this removal step. The way we use, which we feel is quite natural, is to continue removing from the candidate set of agent $i$ an item which agent $i'$ values the most, until the $r$ items in the candidate set of $i$ that are most highly valued by $i'$ are not valued more than $r \cdot \tau$ in total. The pseudo-code of the algorithm is presented below as Algorithm~\ref{alg:matching-removal}; here we use sum-top$_r(S)$ to denote the sum of the $r$ largest elements of $S$ for any \emph{multiset} of real numbers $S$ (or the sum of all elements if $S$ contains less than $r$ elements). The set in line~\ref{step:multiset} of the algorithm is considered as a multiset. The appropriate value of $\tau$ depends on the distribution $\mathcal{U}$ and will be specified later.

\begin{algorithm}[h!]
\caption{Algorithm for any $r\geq 2$}\label{alg:matching-removal}
\begin{algorithmic}[1]
\Procedure{ThresholdMatchingWithRemoval$_\tau$\linebreak$(N, M, \{u_i\}_{i \in [n]})$}{}
\For{$i = 1,2, \dots, n$}
\State $M^*_{\geq \tau}(i) \leftarrow \{j \in M \mid u_i(j) \geq \tau\}$
\For{$i' \in [n] \setminus \{i\}$}
\While{sum-top$_r\left(\{u_{i'}(j) \mid j \in M^*_{\geq \tau}(i)\}\right) > r \cdot \tau$} \label{step:multiset}
\State $M^{*}_{\geq \tau}(i) \leftarrow M^{*}_{\geq \tau}(i) \setminus \argmax_{j \in M^{*}_{\geq \tau}(i)} u_{i'}(j)$ \label{step:remove}
\EndWhile
\EndFor
\EndFor
\State Let $G^*_{\geq \tau} = (N, M, E^*_{\geq \tau})$ denote the graph where $(i, j) \in E^*_{\geq \tau}$ iff $j \in M^*_{\geq \tau}(i)$.
\If{$G^*_{\geq \tau}$ contains a perfect $r$-matching}
\State \Return any perfect $r$-matching of $G^*_{\geq \tau}$
\Else
\State \Return NULL
\EndIf
\EndProcedure
\end{algorithmic}
\end{algorithm}

The above modification ensures that if Algorithm~\ref{alg:matching-removal} returns an allocation, it must be envy-free. Indeed, each agent $i$ has utility at least $\tau$ for every item assigned to her in the $r$-matching, so her total utility is at least $r\cdot\tau$. On the other hand, by construction of the graph $G^*_{\geq \tau}$, each agent $i'$ values the $r$ items assigned to agent $i$ at most $r \cdot \tau$. Thus, the output allocation must be envy-free.

In order to establish Theorem~\ref{thm:existence}, it therefore remains to show that with an appropriate choice of $\tau$, a perfect $r$-matching in $G^*_{\geq \tau}$ exists with high probability. Recall our assumption that the distribution $\mathcal{U}$ from which the utilities are drawn is polynomially bounded at 1. Let $\underline{\theta}, \overline{\theta}, q > 0$ be the associated parameters. It suffices to prove the following lemma:

\begin{lemma}
\label{lem:exists}
Set $\tau := 1-\left(\frac{64\log m}{\underline{\theta}n}\right)^{1/q}$ in Algorithm~\ref{alg:matching-removal}, and let $2\leq r\leq e^{n^{0.1}}$. Then, with high probability, the graph $G^{*}_{\geq \tau}$ contains a perfect $r$-matching.
\end{lemma}

Note that the condition $r\leq e^{n^{0.1}}$ implies that $\tau>0$ for large enough $n$.

The proof of Lemma~\ref{lem:exists} consists of two parts. First, in Section~\ref{sec:removed-bound}, we show that only few edges are removed in line~\ref{step:remove} of Algorithm~\ref{alg:matching-removal}; in particular, we show that with high probability, at most two edges adjacent to any particular vertex are removed. Then, in Section~\ref{sec:resilience}, we show that the existence of a perfect $r$-matching is \emph{locally resilient} (see, e.g.,~\citep{SudakovVu08}) in the following sense: even if we remove a low-degree subgraph from a random graph sampled from the Erd{\H{o}}s-R{\'{e}}nyi random bipartite graph distribution with sufficiently large probability, then the remaining graph still contains a perfect $r$-matching with high probability. Putting these two parts together yields Lemma~\ref{lem:exists}; this is done in Section~\ref{sec:proof-lem-exists}. 

Before we proceed to proving Lemma~\ref{lem:exists}, we perform some preliminary calculations. Picking $\alpha = 1-\tau$ in Definition~\ref{def:polybounded}, we have $\Pr_{u\sim\mathcal{U}}[u > \tau] \geq \frac{64\log m}{n}$. On the other hand, writing $\tau':=3\tau-2$ and letting $\alpha = 1-\tau'$ in Definition~\ref{def:polybounded} yields $\Pr_{u\sim\mathcal{U}}[u > \tau'] \leq \overline{\theta}(3(1-\tau))^q = \frac{C\log m}{n}$ for the constant $C:=3^q\cdot \frac{64\overline{\theta}}{\underline{\theta}}$.

\subsection{Bounding the Number of Edges Removed}
\label{sec:removed-bound}

Let $E_{\geq \tau}$ and $E^*_{\geq \tau}$ denote the set of edges as defined in Algorithm~\ref{alg:matching-basic} and \ref{alg:matching-removal}, respectively. The main result of this subsection is the following lemma:

\begin{lemma} 
\label{lem:removed-bound}
With high probability, the graph $(N, M, E_{\geq \tau} \setminus E^*_{\geq \tau})$ has maximum degree at most 2.
\end{lemma}

In addition to the graphs $G_{\geq \tau}$ and $G^*_{\geq \tau}$ (as defined in Algorithm~\ref{alg:matching-basic} and~\ref{alg:matching-removal} respectively), we consider the graph $G_{> \tau'} = (N, M, E_{> \tau'})$ which can be defined analogously. That is, the neighbor set of $i \in N$ in $G_{> \tau'}$ is $M_{> \tau'}(i) := \{j \in M \mid u_i(j) > \tau'\}$. 

The next proposition states that, for any edge $(i, j)$ that is removed in line~\ref{step:remove} of Algorithm~\ref{alg:matching-removal}, the edge $(i, j)$ must be part of a complete bipartite subgraph $K_{2, \lfloor 2r/3\rfloor + 1}$ of the graph $G_{> \tau'}$.\footnote{The notation $K_{a,b}$ refers to a complete bipartite graph with left and right vertex sets of size $a$ and $b$, respectively.} We note that this is similar to the argument for the simplified algorithm in Section~\ref{sec:simplified}, which, in the new language, states that any edge $(i, j)$ that would be removed in line~\ref{step:remove} of Algorithm~\ref{alg:matching-removal} must be part of a complete bipartite subgraph $K_{2, r}$ of the graph $G_{> \tau'}$, where the threshold $\tau'$ there is chosen (differently) to be $1 - 2r\log n/n$.

\begin{proposition} 
\label{prop:biclique}
If $(i, j) \in E_{\geq \tau} \setminus E^*_{\geq \tau}$, then there exists $i'$ such that $|M_{> \tau'}(i) \cap M_{> \tau'}(i')| > 2r/3$ and $j \in M_{> \tau'}(i) \cap M_{> \tau'}(i')$. 
\end{proposition}

\begin{proof}
First, let us argue that $j \in M_{> \tau'}(i) \cap M_{> \tau'}(i')$. The assumption that $(i, j) \in E_{\geq \tau}$ immediately implies that $j \in M_{>\tau'}(i)$ since $\tau'=3\tau-2<\tau$. Let $i'\in N$ be the vertex in line~\ref{step:multiset} of Algorithm~\ref{alg:matching-removal} that causes the removal of the edge $(i,j)$. At this line, we have sum-top$_r\left(\{u_{i'}(j') \mid j' \in M^*_{\geq \tau}(i)\}\right) > r \cdot \tau$ and $j\in \argmax_{j' \in M^{*}_{\geq \tau}(i)} u_{i'}(j')$. Hence, $u_{i'}(j)\geq (r\cdot\tau)/r = \tau > \tau'$, and therefore $j\in M_{>\tau'}(i')$. We have thus shown that $j\in M_{>\tau'}(i)\cap M_{>\tau'}(i')$.

Next, let $y = \min\{r, |M_{> \tau'}(i) \cap M_{> \tau'}(i')|\}$. The condition sum-top$_r\left(\{u_{i'}(j') \mid j' \in M^*_{\geq \tau}(i)\}\right) > r \cdot \tau$ at the time when the edge $(i,j)$ is removed implies that 
\begin{align*}
r \cdot \tau &< \text{sum-top}_r\left(\{u_{i'}(j') \mid j' \in M^*_{\geq \tau}(i)\}\right) \\
&\leq \text{sum-top}_r\left(\{u_{i'}(j') \mid j' \in M_{> \tau'}(i)\}\right) \\
&\leq y\cdot 1 + (r - y)\cdot \tau' \\
&= y + (r - y)(3\tau - 2) \\
&= (3 - 3\tau)y + (3\tau - 2)r,
\end{align*}
which implies that $y > 2r/3$, as claimed.
\end{proof}

We now use Proposition~\ref{prop:biclique} to prove Lemma~\ref{lem:removed-bound}. We consider two cases: $r \geq 3$ and $r = 2$.

\subsubsection{The Case $r\geq 3$.}

The proof for the case $r \geq 3$ is similar to that of the simplified algorithm (Section~\ref{sec:simplified}). In particular, we show that with high probability, no edge is removed in Algorithm~\ref{alg:matching-removal}. This also means that Algorithms~\ref{alg:matching-basic} and \ref{alg:matching-removal} are equivalent with high probability.

\begin{proposition} 
\label{prop:r3}
Let $3 \leq r \leq e^{n^{0.1}}$. Then, with high probability, $E_{\geq \tau} = E^*_{\geq \tau}$.
\end{proposition}

\begin{proof}
For convenience, let $p_{> \tau'} := \Pr_{u \sim \mathcal{U}}[u > \tau']$. Note that $p_{> \tau'} \leq \frac{C\log m}{n}\leq C_0/n^{0.9}$ for $C_0:=2C$, where the last inequality holds for sufficiently large $n$. We will argue that with high probability, there are no distinct $i, i' \in N$ such that $|M_{> \tau'}(i) \cap M_{> \tau'}(i')| > 2r/3$. Together with Proposition~\ref{prop:biclique}, this implies that no edge is removed and therefore $E_{\geq \tau} = E^*_{\geq \tau}$.

To show this, we use the standard first moment method. Fix distinct $i, i' \in N$ and a subset $S \subseteq M$ of size $x := \lfloor 2r/3 \rfloor + 1$. The probability that $S \subseteq M_{> \tau'}(i) \cap M_{> \tau'}(i')$ is exactly $(p_{>\tau'})^{2x}$. Hence, by taking the union bound over all choices of $i, i'$ and $S$, the probability that $|M_{> \tau'}(i) \cap M_{> \tau'}(i')| > 2r/3$ for some $i,i'$ is at most
\begin{align*}
n^2 \binom{m}{x} (p_{> \tau'})^{2x} &\leq n^2 \left(\frac{em}{x}\right)^x (p_{> \tau'})^{2x} \\
(\text{since } x \geq 3) &\leq (n(p_{>\tau'})^{1.2})^2 \left(\frac{em(p_{>\tau'})^{1.2}}{x}\right)^x \\
(\text{since } x > 2r/3) &< (n(p_{>\tau'})^{1.2})^2 \left(1.5en(p_{>\tau'})^{1.2}\right)^x \\
(\text{since } p_{>\tau'} \leq \frac{C_0}{n^{0.9}}) &\leq (C_0^{1.2}n^{-0.08})^2(1.5eC_0^{1.2}n^{-0.08})^x \\
&= o(1), 
\end{align*}
which concludes the proof.
\end{proof}

\subsubsection{The Case $r=2$.}

As argued earlier, in the case $r=2$, some edges must be removed in order to guarantee that the output allocation is envy-free. The following proposition ensures that with high probability, for any vertex, at most two edges adjacent to it are removed in Algorithm~\ref{alg:matching-removal}.

\begin{proposition} 
\label{prop:r2}
Let $r = 2$. Then, with high probability, the graph $(N, M, E_{\geq \tau} \setminus E^*_{\geq \tau})$ has maximum degree at most 2.
\end{proposition}

\begin{proof}
Observe that, for $r = 2$, Proposition~\ref{prop:biclique} can be restated as follows: if $(i, j) \in E_{\geq \tau} \setminus E^*_{\geq \tau}$, then there exist $i' \in N$ and $j' \in M$ such that $i, i', j, j'$ form a complete bipartite graph $K_{2, 2}$ in the graph $G_{> \tau'}$.

Now, suppose that some vertex $u \in N \cup M$ appears in at least three edges in $E_{\geq \tau} \setminus E^*_{\geq \tau}$. The previous paragraph implies each such edge must be contained in a copy of $K_{2, 2}$ in the graph $G_{> \tau'}$. Since the three edges from $u$ are distinct, not all three of these copies can be identical. As a result, $u$ must be contained in two different copies of $K_{2,2}$, which means that at least one of the graphs shown in Figure~\ref{fig:biclique-union} must appear as a subgraph of $G_{> \tau'}$. Notice that by the union bound, for any graph $H = (V_H, E_H)$, the probability that it appears as a subgraph of $G_{> \tau'}$ is at most $(n + m)^{|V_H|}(p_{> \tau'})^{|E_H|} \leq (3n)^{|V_H|} \left(\frac{C \log m}{n}\right)^{|E_H|}$. However, all graphs $H$ in Figure~\ref{fig:biclique-union} satisfy $|E_H| \geq |V_H| + 1$. Hence, the probability that each of them appears as a subgraph is at most $\frac{(C\log n)^{O(1)}}{n} = o(1)$. Using the union bound, the probability that at least one of these graphs appears as a subgraph of $G_{> \tau'}$ is also $o(1)$. This implies that the probability that at least one of the vertices is adjacent to more than two edges in $E_{\geq \tau} \setminus E^*_{\geq \tau}$ is $o(1)$, as desired.
\end{proof}

\begin{figure*}
\centering
\begin{subfigure}[m]{0.18\textwidth}
\begin{tikzpicture}[scale=0.75]
    \node[shape=circle,draw=black,minimum size=20pt,fill=light-gray,ultra thick] (A) at (0,0) {u};
    \node[shape=circle,draw=black,minimum size=20pt,ultra thick] (B) at (0,1.5) {};
    \node[shape=circle,draw=black,minimum size=20pt,ultra thick] (C) at (2,2) {};
    \node[shape=circle,draw=black,minimum size=20pt,ultra thick] (D) at (2,0.5) {};
    \node[shape=circle,draw=black,minimum size=20pt,fill=light-gray] (E) at (0,-1.5) {};
    \node[shape=circle,draw=black,minimum size=20pt,fill=light-gray] (F) at (2,-0.5) {};
    \node[shape=circle,draw=black,minimum size=20pt,fill=light-gray] (G) at (2,-2) {};

    \path [-] (A) edge (C);  
    \path [-] (B) edge (C);
    \path [-] (A) edge (D);
    \path [-] (B) edge (D);      
    \path [-] (A) edge (F); 
    \path [-] (A) edge (G);
    \path [-] (E) edge (F);
    \path [-] (E) edge (G);   
\end{tikzpicture}
\end{subfigure}
\begin{subfigure}[m]{0.18\textwidth}
\begin{tikzpicture}[scale=0.75]
    \node[shape=circle,draw=black,minimum size=20pt,fill=light-gray,ultra thick] (A) at (0,0) {u};
    \node[shape=circle,draw=black,minimum size=20pt,ultra thick] (B) at (0,2) {};
    \node[shape=circle,draw=black,minimum size=20pt,ultra thick] (C) at (2,2) {};
    \node[shape=circle,draw=black,minimum size=20pt,fill=light-gray,ultra thick] (D) at (2,0) {};
    \node[shape=circle,draw=black,minimum size=20pt,fill=light-gray] (E) at (0,-2) {};
    \node[shape=circle,draw=black,minimum size=20pt,fill=light-gray] (G) at (2,-2) {};

    \path [-] (A) edge (C);  
    \path [-] (B) edge (C);
    \path [-] (A) edge (D);
    \path [-] (B) edge (D);      
    \path [-] (A) edge (G);
    \path [-] (E) edge (G);   
    \path [-] (E) edge (D);
\end{tikzpicture}
\end{subfigure}
\begin{subfigure}[m]{0.18\textwidth}
\begin{tikzpicture}[scale=0.75]
    \node[shape=circle,draw=black,minimum size=20pt,fill=light-gray,ultra thick] (A) at (0,0) {u};
    \node[shape=circle,draw=black,minimum size=20pt,ultra thick] (B) at (0,2) {};
    \node[shape=circle,draw=black,minimum size=20pt,ultra thick,fill=light-gray] (C) at (2,1) {};
    \node[shape=circle,draw=black,minimum size=20pt,fill=light-gray,ultra thick] (D) at (2,-1) {};
    \node[shape=circle,draw=black,minimum size=20pt,fill=light-gray] (E) at (0,-2) {};

    \path [-] (A) edge (C);  
    \path [-] (B) edge (C);
    \path [-] (A) edge (D);
    \path [-] (B) edge (D);      
    \path [-] (E) edge (C);   
    \path [-] (E) edge (D);
\end{tikzpicture}
\end{subfigure}
\begin{subfigure}[m]{0.18\textwidth}
\begin{tikzpicture}[scale=0.75]
    \node[shape=circle,draw=black,minimum size=20pt,fill=light-gray,ultra thick] (A) at (0,1) {u};
    \node[shape=circle,draw=black,minimum size=20pt,ultra thick] (C) at (2,2) {};
    \node[shape=circle,draw=black,minimum size=20pt,ultra thick] (D) at (2,0.5) {};
    \node[shape=circle,draw=black,minimum size=20pt,fill=light-gray,ultra thick] (E) at (0,-1) {};
    \node[shape=circle,draw=black,minimum size=20pt,fill=light-gray] (F) at (2,-0.5) {};
    \node[shape=circle,draw=black,minimum size=20pt,fill=light-gray] (G) at (2,-2) {};

    \path [-] (A) edge (C);  
    \path [-] (A) edge (D);    
    \path [-] (A) edge (F); 
    \path [-] (A) edge (G);
    \path [-] (E) edge (F);
    \path [-] (E) edge (G);
    \path [-] (E) edge (C);
    \path [-] (E) edge (D);   
\end{tikzpicture}
\end{subfigure}
\begin{subfigure}[m]{0.18\textwidth}
\begin{tikzpicture}[scale=0.75]
    \node[shape=circle,draw=black,minimum size=20pt,fill=light-gray,ultra thick] (A) at (0,1) {u};
    \node[shape=circle,draw=black,minimum size=20pt,ultra thick] (C) at (2,2) {};
    \node[shape=circle,draw=black,minimum size=20pt,ultra thick,fill=light-gray] (D) at (2,0) {};
    \node[shape=circle,draw=black,minimum size=20pt,fill=light-gray,ultra thick] (E) at (0,-1) {};
    \node[shape=circle,draw=black,minimum size=20pt,fill=light-gray] (G) at (2,-2) {};

    \path [-] (A) edge (C);  
    \path [-] (A) edge (D);    
    \path [-] (A) edge (G);
    \path [-] (E) edge (G);
    \path [-] (E) edge (C);
    \path [-] (E) edge (D);   
\end{tikzpicture}
\end{subfigure}
\caption{All possible unions of two distinct complete bipartite graphs $K_{2, 2}$ which share at least one vertex $u$, up to isomorphism. The shaded vertices constitute one copy of $K_{2, 2}$ whereas the thickened vertices constitute another.} 
\label{fig:biclique-union}
\end{figure*}
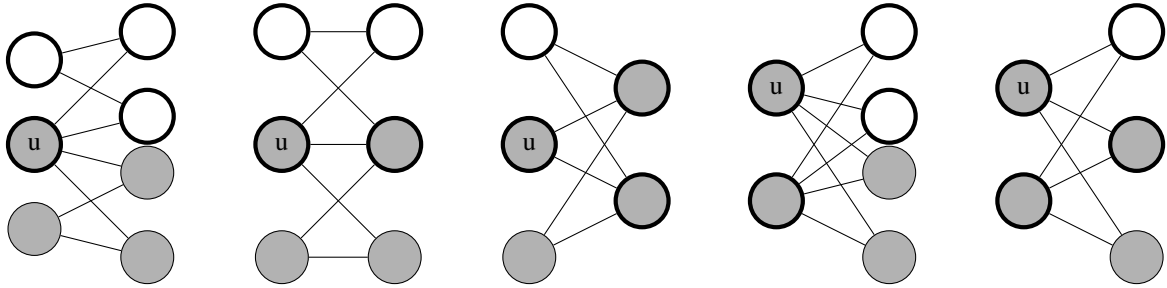

Finally, we note that Propositions~\ref{prop:r3} and~\ref{prop:r2} together imply Lemma~\ref{lem:removed-bound}.

\subsection{Local Resilience of Perfect $r$-Matching}
\label{sec:resilience}

In this subsection, we show that in a random bipartite graph sampled from the Erd{\H{o}}s-R{\'{e}}nyi random bipartite graph distribution $\mathcal{G}(n, rn, p)$  with sufficiently large $p$, not only does a perfect $r$-matching exist, but the existence is also robust in the following sense: even if we remove edges from the graph, as long as not too many edges adjacent to each vertex are removed, a perfect $r$-matching still exists. Such ``robustness'' is known in the literature as \emph{local resilience}. In particular, the local resilience of perfect matchings was shown by \cite{SudakovVu08}. We will extend their proof to the case of perfect $r$-matchings. However, we note that our bound will be slightly weaker than theirs, since our main goal is to derive a bound that is sufficient for the algorithm to work and not to find the best possible parameters.

A typical method for establishing the existence of a perfect matching, which was used both by \cite{SudakovVu08} and by \cite{ErdosRe64}, is to show that the graph satisfies the condition of Hall's Marriage Theorem. For any graph $G$ and any set $S$ of vertices in $G$, denote by $N_G(S)$ the set of vertices adjacent to at least one vertex in~$S$.

\begin{proposition}[Hall's Marriage Theorem]
Let $G = (A, B, E)$ be any bipartite graph such that $|A| = |B|$. If $|N_G(S)| \geq |S|$ for all subsets $S \subseteq A$, then $G$ has a perfect matching.
\end{proposition}

Recall that $G = (A, B, E)$ has a perfect $r$-matching if and only if the graph $(A \times [r], B, E')$, where $((a, \ell), b) \in E'$ iff $(a, b) \in E$, has a perfect matching. Hence, Hall's Marriage Theorem immediately extends to $r$-matchings:

\begin{proposition} \label{prop:hall-rmatching}
Let $G = (A, B, E)$ be any bipartite graph such that $|B| = r|A|$. If $|N_G(S)| \geq r|S|$ for all subsets $S \subseteq A$, then $G$ has a perfect $r$-matching.
\end{proposition}

One way to show that the condition of Hall's Marriage Theorem is satisfied is to show that there is at least one edge between any sets $S \subseteq A$ and $T \subseteq B$ of appropriate sizes. To ensure that the existence of a perfect $r$-matching is locally resilient, we need to show not only that one edge exists, but also that many edges exist. This can be done via standard concentration bounds. For any graph $G$ and any sets $S,T$ of vertices in $G$, denote by $E_G(S,T)$ the set of edges connecting a vertex in $S$ to a vertex in $T$.

\begin{lemma} 
\label{lem:exp-edge}
Let $G = (A, B, E)$ be a graph sampled from the Erd{\H{o}}s-R{\'{e}}nyi random bipartite graph distribution $\mathcal{G}(n,m,p)$ with $p \geq \frac{64\log m}{n}$. Then, with high probability, the following holds for all subsets $S \subseteq A$ and $T \subseteq B$ such that $|T| = m - r|S| + 1$:
\begin{align} \label{eq:edge-concen}
|E_G(S, T)| > \left(16\log m\right) \cdot \min\{|S|, |T|\}.
\end{align}
\end{lemma}

The proof of Lemma~\ref{lem:exp-edge} can be found in Appendix~\ref{app:exp-edge}.

With Lemma~\ref{lem:exp-edge} ready, we now establish the local resilience of the existence of perfect $r$-matchings in random graphs.

\begin{lemma} 
\label{lem:matching-robust}
Let $G = (A, B, E)$ be a graph sampled from the Erd{\H{o}}s-R{\'{e}}nyi random bipartite graph distribution $\mathcal{G}(n,m,p)$ with $p \geq \frac{64\log m}{n}$. Then, with high probability, for any subgraph $H = (A, B, E')$ of $G$ with maximum degree at most $16\log m$, the graph $G - H = (A, B, E \setminus E')$ contains a perfect $r$-matching.
\end{lemma}

\begin{proof}
From Lemma~\ref{lem:exp-edge}, with high probability, \eqref{eq:edge-concen} holds for all $S \subseteq A, T \subseteq B$ with $|T| = m - r|S| + 1$. We claim that this implies that $G - H = (A, B, E \setminus E')$ contains a perfect $r$-matching. Suppose for the sake of contradiction that $G - H$ does not contain an $r$-perfect matching. Proposition~\ref{prop:hall-rmatching} implies that there exists a set $S\subseteq A$ such that $|N_{G-H}(S)| \leq r|S| - 1$. Let $T$ be any subset of $B \setminus N_{G-H}(S)$ of size $m - r|S| + 1$. Since $T \cap N_{G-H}(S) = \emptyset$, we have $E_{G-H}(S, T) = \emptyset$. However, by~\eqref{eq:edge-concen}, $|E_G(S, T)| > (16 \log m) \cdot \min\{|S|, |T|\}$. This means that at least one vertex in $S \cup T$ has degree more than $16\log m$ in $H$, which is a contradiction.
\end{proof}

\subsection{Putting Things Together} 
\label{sec:proof-lem-exists}

With Lemmas~\ref{lem:removed-bound} and \ref{lem:matching-robust} in hand, we can (finally) prove Lemma~\ref{lem:exists}.

\begin{proof}[Proof of Lemma~\ref{lem:exists}]
First, Lemma~\ref{lem:removed-bound} ensures that with high probability, for each vertex, at most two edges adjacent to it are removed from $G_{\geq\tau}$ in Algorithm~\ref{alg:matching-removal}, where $G_{\geq\tau}$ is defined as in Algorithm~\ref{alg:matching-basic}. Recall also that $G_{\geq \tau}$ is distributed according to the Erd{\H{o}}s-R{\'{e}}nyi random bipartite graph distribution $\mathcal{G}(n,m,p)$ with $p=\Pr_{u\sim\mathcal{U}}[u\geq\tau]\geq\frac{64\log m}{n}$. It therefore follows from Lemma~\ref{lem:matching-robust} that a perfect $r$-matching exists in $G^*_{\geq \tau}$ with high probability.
\end{proof}

\section{Non-Existence}
\label{sec:non-existence}

Our main non-existence result states that envy-free allocations are unlikely to exist if $m=O(n\log n/\log\log n)$ is not ``close to'' being a multiple of $n$. This improves upon the $m=n+o(n)$ lower bound of \cite{DickersonGoKa14} and comes to within a $\Theta(\log\log n)$ factor of matching their upper bound.

\begin{theorem}
\label{thm:non-existence}
For any real numbers $\theta > 0$, $\epsilon\in(0,1)$, and $q \geq 1$, there exists $c > 0$ depending only on $\theta, \epsilon, q$ such that the following holds: For any positive integer $r \leq \frac{c \log n}{\log \log n}$, if $m \in [rn + n^\epsilon, (r + 1)n - n^\epsilon]$ and $\mathcal{U}$ is $(\theta, q)$-polynomially bounded below at 1, then, with high probability, there is no envy-free allocation.
\end{theorem}

We remark that since we only require the distribution to be polynomially bounded below, the assumption $q\geq 1$ does not introduce a loss of generality. Next, we give an overview of the proof of Theorem~\ref{thm:non-existence}; the full proof can be found in Appendix~\ref{app:non-existence}.

The proof is based on the first moment method; the key is to show that for any fixed allocation, the probability (over the random utilities drawn) that it is envy-free is $\ll 1/n^m$. Since there are $n^m$ possible allocations, the union bound implies that with high probability, no envy-free allocation exists.

To give an intuition for this bound, let us consider a simplified setting where $m = (r + 0.5)n$ and the distribution $\mathcal{U}$ is uniform on $[0, 1]$. Intuitively, the ``more balanced'' the allocation is, the harder it is to bound the probability that the allocation is envy-free. Following this intuition, let us consider the ``most balanced'' allocation where $0.5n$ agents receive $r + 1$ items and the remaining agents receive $r$ items. The key observation is that, for the allocation to be envy-free for every agent in the latter group, any such agent must have utility at most $r$ for the $r+1$ items in the bundle of any agent in the first group. For a fixed agent in the second group and a fixed agent in the first group, this happens with probability at most $1 - 1/(r + 1)^{r+1}$. Indeed, if each of the $r+1$ items yields utility at least $r/(r+1)$ to the agent, the requirement is not satisfied. Now, since there are $0.25n^2$ such pairs of agents, the probability that this fixed allocation is envy-free is at most $\left(1 - \frac{1}{(r + 1)^{r+1}}\right)^{0.25n^2} = \exp\left(\Theta\left(\frac{-n^2}{(r + 1)^{r + 1}}\right)\right)$. Hence, as long as $r \ll \log n/\log \log n$, this term is at most, say, $\exp(-n^{1.9})$, which is indeed much smaller than $n^{-m}$.

The full proof proceeds along the lines of the argument above, but we need to be more careful as we must also deal with other ``less balanced'' allocations.

\section{Discussion}

In this paper, we study the existence and non-existence of envy-free allocations and essentially close the gap left open by \cite{DickersonGoKa14} with regard to the transition between the two phases. On the positive side, we show that if the number of items is a multiple of the number of agents, an envy-free allocation is likely to exist as long as the former quantity is at least twice the latter. On the negative side, we show that if the number of items is not ``close to'' being a multiple of the number of agents, an envy-free allocation is unlikely to exist even when the former quantity exceeds the latter by almost a logarithmic factor. Both of our results make use of several new ideas that may be useful for other problems in fair division.

As we mentioned earlier, all of the works on the asymptotic existence of fair allocations thus far have assumed that agents are endowed with additive utilities. While additivity provides a reasonable trade-off between simplicity and expressiveness, it would be interesting to establish analogous results that hold for more general classes of utilities. Going beyond additivity introduces several complications; for example, the welfare-maximizing allocation is no longer simply the one that assigns every item to the agent who values it most, and giving an agent several goods that she values highly does not guarantee that the agent will also have a correspondingly high value for the whole bundle. Nevertheless, a starting point may be to prove results for specific distributions over utilities from a well-structured class such as that of submodular valuations.

Another possible avenue for future work is to consider the setting where instead of allocating items to individual agents, we divide them among \emph{groups} of agents~\citep{Suksompong18-2}. The agents in each group share the same set of items but may have different preferences. This is the case, for example, when dividing household goods among families or resources between departments in a university.  
\cite{ManurangsiSu17} generalized the results of \cite{DickersonGoKa14} to the group setting and left a logarithmic gap between existence and non-existence. We are hopeful that the techniques we introduce in the present work will help towards closing this gap as well.

\section*{Acknowledgments}

This work was partially supported by NSF Awards CCF-1655215, CCF-1813188, CCF-1815434, and by a Stanford Graduate Fellowship. We would like to thank the anonymous reviewers for their helpful comments.

\bibliography{envyfree}
\bibliographystyle{aaai}

\newpage

\appendix

\section*{\centering Appendix:\\ When Do Envy-Free Allocations Exist?}
\bigskip

\section{Proof of Lemma~\ref{lem:exp-edge}}
\label{app:exp-edge}

To prove Lemma~\ref{lem:exp-edge}, we will use the Chernoff bound, which is stated for convenience below.

\begin{proposition}[Chernoff bound] \label{prop:chernoff}
Let $X_1,X_2, \dots, X_r$ be i.i.d. random variables that take on values in the interval $[0, 1]$, and let $X:=X_1 + \dots + X_r$. For every $\delta \geq 0$, we have
\begin{align*}
\Pr[X \leq (1 - \delta)\E[X]] &\leq \exp{\left(\frac{-\delta^2\E[X]}{2}\right)}.
\end{align*}
\end{proposition}

\begin{proof}[Proof of Lemma~\ref{lem:exp-edge}]
If $S=\emptyset$, we must have $|T|=m+1$, which is impossible since $|B|=m$. Fix a subset $\emptyset \ne S \subseteq A$ and $T \subseteq B$ such that $|T| = m - rs + 1$, where $s:=|S|$. We will compute the probability that $S, T$ violate~\eqref{eq:edge-concen}. 

Observe that $|E_G(S, T)|$ is simply a sum of $|S| |T|$ i.i.d. Bernoulli random variables that take on the value 1 with probability $p$. Hence, by Proposition~\ref{prop:chernoff}, we have
\begin{align} \label{eq:edge-chernoff}
\Pr\left[|E_G(S, T)| \leq \frac{p|S||T|}{2}\right] \leq \exp\left(-\frac{p|S||T|}{8}\right).
\end{align}

Now, observe that 
\begin{align} 
p|S||T| &= p \cdot \min\{|S|, |T|\} \cdot \max\{|S|, |T|\} \nonumber \\
&\geq p \cdot \min\{|S|, |T|\} \cdot \frac{r|S| + |T|}{r + 1} \nonumber \\
&= p \cdot \min\{|S|, |T|\} \cdot \frac{m + 1}{r + 1} \nonumber \\
&\geq p \cdot \min\{|S|, |T|\} \cdot n/2 \nonumber \\
&\geq \left(32\log m\right) \cdot \min\{|S|, |T|\} \label{eq:expected-edges}.
\end{align}
Thus, by combining~\eqref{eq:edge-chernoff} and~\eqref{eq:expected-edges}, we get
\begin{align*}
&\Pr\left[|E_G(S, T)| \leq \left(16\log m\right) \cdot \min\{|S|, |T|\}\right] \\
&\leq \exp\left(-\frac{p|S||T|}{8}\right).
\end{align*}

By the union bound, the probability that there exist $S, T$ that violate~\eqref{eq:edge-concen} is at most
\begin{align*}
&\sum_{s=1}^n \exp\left(\frac{-p s (m -rs + 1)}{8}\right) \binom{n}{s} \binom{m}{m - rs + 1} \\
&\leq \sum_{s=1}^n \exp\left(\frac{-p s (m -rs + 1)}{8}\right) n^{\min\{s, n - s\}}\\
& \quad\cdot m^{\min\{rs - 1, m - rs + 1\}} \\
&= \sum_{s=1}^n \exp\biggl(\frac{-p s (m -rs + 1)}{8} + \log n \cdot \min\{s, n - s\} \\
& \quad+ \log m \cdot \min\{rs - 1, m - rs + 1\}\biggr) \\
&\leq \sum_{s=1}^n \exp\biggl(\frac{-p s (m -rs + 1)}{8} \\
&\quad+ 2\log m \cdot \min\{rs, m - rs + 1\}\biggr) \\
&= \sum_{s=1}^n \exp\biggl(\left(\frac{-(p/r) \max\{rs, m -rs + 1\}}{8} + 2\log m\right) \\
&\quad \cdot \min\{rs, m - rs + 1\}\biggr) \\
&\leq \sum_{s=1}^n \exp\biggl(\left(\frac{-(p/r) \cdot (m/2)}{8} + 2\log m\right) \\
&\quad \cdot \min\{rs, m - rs + 1\}\biggr) \\
&\leq \sum_{s=1}^n \exp\biggl(\left(-4\log m + 2\log m\right) \cdot \min\{rs, m - rs + 1\}\biggr) \\
&\leq \sum_{s=1}^n \exp\left(-2\log m\right) \\
&\leq 1/n,
\end{align*}
which concludes the proof.
\end{proof}

\section{Proof of Theorem~\ref{thm:non-existence}}
\label{app:non-existence}

Let $c = 0.1\epsilon\theta/q$. We will show that, for any sufficiently large $n$ and any allocation $\phi: M \to N$,
\begin{align} \label{eq:one-allocation}
\Pr[\phi \text{ is envy-free}] \leq n^{-2m}.
\end{align}
Since there are $n^m$ different allocations, the union bound implies that the probability that an envy-free allocation exists is at most $n^{-m} = o(1)$. Hence, it suffices for us to show \eqref{eq:one-allocation}.

To prove \eqref{eq:one-allocation}, observe that for any fixed $\phi$, the probability that $\phi$ is envy-free for each agent is independent. That is,
\begin{align}
&\Pr[\phi \text{ is envy-free}] \nonumber \\
&= \prod_{i \in N} \Pr[\phi \text{ is envy-free for agent } i]. \label{eq:all-agent}
\end{align}

Let $z_i = |\phi^{-1}(i)|$ denote the number of items that agent $i$ receives. Since agent $i$ values her items at most $z_i$ in total, if $\phi$ is envy-free for $i$, then $i$ must value items allocated to any other agent $i'$ at most $z_i$. More precisely, this implies that
\begin{align}
&\Pr[\phi \text{ is envy-free for agent } i] \nonumber \\
 &\leq \Pr\left[\forall i' \in N, \sum_{j \in \phi^{-1}(i')} u_i(j) \leq z_i\right] \nonumber \\
&= \prod_{i' \in N} \Pr\left[\sum_{j \in \phi^{-1}(i')} u_i(j) \leq z_i\right]. \label{eq:one-agent}
\end{align} 

Next, we use the assumption that $\mathcal{U}$ is $(\theta, q)$-polynomially bounded below at 1 to derive an upper bound for the probability that $\sum_{j \in \phi^{-1}(i')} u_i(j) \leq z_i$, for $i$ and $i'$ such that $z_i \leq r < z_{i'}$. 

\begin{lemma}
\label{lem:compare-sum}
Let $\rho:=\sqrt[4]{1 - \left(\theta \left(\frac{1}{r + 1}\right)^q\right)^{r + 1}}$. For every $i, i' \in N$ such that $z_i \leq r < z_{i'}$, we have
\begin{align}
\Pr\left[\sum_{j \in \phi^{-1}(i')} u_i(j) \leq z_i\right] \leq \rho^{z_{i'}/r}. \nonumber
\end{align}
\end{lemma}

Note that since $\mathcal{U}$ is $(\theta,q)$-polynomially bounded below at 1, we must have $\theta\leq 1$, and therefore $\left(\theta \left(\frac{1}{r + 1}\right)^q\right)^{r + 1}\leq 1$.

\begin{proof}
To prove this bound, let $\gamma$ be the probability that the sum of $r + 1$ i.i.d. random variables drawn from $\mathcal{U}$ exceeds $r$. This probability is at least the probability that all of the $r + 1$ random variables are of value more than $r/(r + 1)$. Thus, we have
\begin{align*}
\gamma &\geq \left(\Pr_{u \sim \mathcal{U}}\left[u > \frac{r}{r+1}\right]\right)^{r + 1} \\
&\geq \left(\theta \left(\frac{1}{r + 1}\right)^q\right)^{r + 1} = 1 - \rho^4
\end{align*}
where the second inequality follows from our assumption on $\mathcal{U}$.

Let $\zeta := \left\lfloor \frac{z_{i'}}{r + 1}\right\rfloor\geq 1$. We next partition the items assigned to $i'$ into $\zeta$ parts each of size at least $r + 1$, i.e., let $\phi^{-1}(i') = S_1 \cup \cdots \cup S_{\zeta}$, where the sets $S_t$ are all disjoint and contain at least $r + 1$ items each. Observe that the probabilities that the sum in each part exceeds $r$ are independent, which means that
\begin{align*}
&\Pr\left[\sum_{j \in \phi^{-1}(i')} u_i(j) \leq z_i\right] \\
&\leq \Pr\left[\sum_{j \in \phi^{-1}(i')} u_i(j) \leq r\right] \\
&\leq \Pr\left[\forall t \in \left[\zeta\right], \sum_{j \in S_t} u_i(j) \leq r\right] \\
&= \prod_{t = 1}^{\zeta} \Pr\left[\sum_{j \in S_t} u_i(j) \leq r\right].
\end{align*} 
Now, note that for any $t=1,2,\dots,\zeta$, the probability that $\sum_{j \in S_t} u_i(j) \leq r$ is at most the probability that the sum of $r + 1$ i.i.d. random variables sampled according to $\mathcal{U}$ is at most $r$. Hence, by definition of $\gamma, \rho$ and the bound established above, we have
\begin{align*}
\Pr\left[\sum_{j \in \phi^{-1}(i')} u_i(j) \leq z_i\right] &\leq \left(1 - \gamma\right)^{\zeta} \\
&\leq (1 - \gamma)^{\frac{z_{i'}}{4r}} \leq \rho^{z_{i'}/r},
\end{align*}
as desired.
\end{proof}

Combining Lemma~\ref{lem:compare-sum} with~\eqref{eq:all-agent} and~\eqref{eq:one-agent} yields
\begin{align*}
\Pr[\phi \text{ is envy-free}] &\leq \prod_{i \in N \atop z_i \leq r} \prod_{i' \in N \atop z_{i'} > r} \rho^{z_{i'}/r} \\
&= \prod_{i \in N \atop z_i \leq r} \rho^{\left(\sum_{i' \in N \text{ s.t. } z_{i'} > r} z_{i'}\right)/r}.
\end{align*}
Let $\ell := m - rn$ be the remainder upon dividing $m$ by $n$. The sum in the exponent can be rearranged as
\begin{align*}
\sum_{i' \in N \atop z_{i'} > r} z_{i'} = m - \sum_{i' \in N \atop z_{i'} \leq r} z_{i'} \geq m - rn = \ell.
\end{align*} 
Thus, we have
\begin{align*}
\Pr[\phi \text{ is envy-free}] \leq \prod_{i \in N \atop z_i \leq r} \rho^{\ell/r} = \rho^{(\ell/r) \cdot |\{i \in N \mid z_i \leq r\}|}.
\end{align*}
Now, observe that $|\{i \in N \mid z_i \leq r\}|$ is at least $n - \frac{m}{r + 1} = \frac{n - \ell}{r + 1}$, which implies that
\begin{align}
\Pr[\phi \text{ is envy-free}] &\leq \rho^{\frac{\ell (n - \ell)}{r(r + 1)}} \nonumber \\
&\leq \rho^{\frac{n \cdot \min\{\ell, n - \ell\}}{2r(r + 1)}} \leq \rho^{\frac{n^{1+\epsilon}}{2r(r + 1)}}. \label{eq:tmp1}
\end{align}

Finally, we can bound $\rho$ as follows:
\begin{align}
\rho &= \sqrt[4]{1 - \left(\theta \left(\frac{1}{r + 1}\right)^q\right)^{r + 1}} \nonumber \\
&\leq \sqrt[4]{1 - \left(\frac{\theta}{r + 1}\right)^{q(r + 1)}} \nonumber \\
&\leq \sqrt[4]{1 - \left(\frac{\theta}{2r}\right)^{2qr}} \nonumber \\
&\leq \sqrt[4]{1 - \left(\frac{\theta}{2c\log n/\log \log n}\right)^{4qc\log n/\log \log n}} \nonumber \\
&= \sqrt[4]{1 - \left(\frac{q\log\log n}{0.2\epsilon\log n}\right)^{0.4\epsilon\theta\log n/\log \log n}} \nonumber \\
&\leq \sqrt[4]{1 - \left(1/\log n\right)^{0.4\epsilon\theta\log n/\log \log n}} \nonumber \\
&\leq \sqrt[4]{1 - \left(1/\log n\right)^{0.5\epsilon \log n/\log \log n}} \nonumber \\
&= \sqrt[4]{1 - n^{-0.5\epsilon}} \nonumber \\
&\leq e^{-n^{-0.5\epsilon}/4}, \label{eq:rho} 
\end{align}
where the first inequality follows from $\theta\leq 1$ and the last inequality follows from the estimate $1-x\leq e^{-x}$, which holds for any real number $x$.

Combining~\eqref{eq:tmp1} and~\eqref{eq:rho}, we get
\begin{align*}
\Pr[\phi \text{ is envy-free}] &\leq e^{-\frac{n^{1 + 0.5\epsilon}}{8r(r + 1)}} = n^{-\frac{n^{1 + 0.5\epsilon}}{8r(r + 1)\log n}} ,
\end{align*}
which is at most $n^{-2m}$ for sufficiently large $n$. This yields~\eqref{eq:one-allocation} and concludes our proof.

\end{document}